\documentclass[11pt,letter]{article}
\usepackage{multirow}
\usepackage{amssymb}
\usepackage{amsmath}
\usepackage{amsthm}
\usepackage{thm-restate}

\usepackage[pdftex, plainpages = false, pdfpagelabels, 
                 bookmarks=false,
                 bookmarksopen = true,
                 bookmarksnumbered = true,
                 breaklinks = true,
                 linktocpage,
                 pagebackref,
                 colorlinks = true,  
                 linkcolor = blue,
                 urlcolor  = blue,
                 citecolor = red,
                 anchorcolor = green,
                 hyperindex = true,
                 hypertexnames = false,
                 hyperfigures
                 ]{hyperref}
\usepackage{complexity} 
\usepackage{enumitem}
\setlist{nosep}

\usepackage{color}
\usepackage{tikz}

\usetikzlibrary{arrows,shapes,positioning,shadows,trees}

\usepackage{vmargin}
\setmarginsrb{1in}{1in}{1in}{1in}{0pt}{0pt}{0pt}{6mm}

\newcommand{\OO}{\mathcal{O}}

\newtheorem{lemma}{Lemma}

\newtheorem{claim}{Claim}[section]
\newtheorem{observation}[lemma]{Observation}
\newtheorem{proposition}{Proposition}[section]
\newtheorem{reduction rule}{Reduction Rule}[section]
\newtheorem{marking-scheme}{Marking Scheme}[section]

\newcommand{\defproblemout}[3]{
  \vspace{1mm}
\noindent\fbox{
  \begin{minipage}{0.96\textwidth}
  \begin{tabular*}{\textwidth}{@{\extracolsep{\fill}}lr} #1 \\ \end{tabular*}
  {\bf{Input:}} #2  \\
  {\bf{Output:}} #3
  \end{minipage}
  }
  \vspace{1mm}
}
\usepackage{algorithmic}
\usepackage{algorithm}
\usepackage[utf8]{inputenc}
\usepackage[LGR,T1]{fontenc} 
\usepackage[T1]{fontenc}
\usepackage{amsmath}

\newcommand{\oct}{\textnormal{\textsc{Odd Cycle Transversal}}}
\newcommand{\OCT}{\textnormal{\textsc{OCT}}}
\newcommand{\maxbipfull}{\textnormal{\textsc{Odd Cycle Transversal}}}

\newcommand{\w}{\textnormal{\texttt{w}}}

\newcommand{\blob}{\textnormal{\texttt{blob}}}

\title{Odd Cycle Transversal on $P_5$-free Graphs in Polynomial Time\footnote{Funding acknowledgements: Agrawal is supported by SERB Startup Research Grant, no. SRG/2022/000962; Lima is supported by the Independent Research Fund Denmark grant agreement number 2098-00012B; Lokshtanov is supported by NSF grant CCF-2008838; Rzążewski is supported by the European Research Council (ERC) under the European Union’s Horizon 2020 research and innovation programme grant agreement number 948057; Saurabh is supported by the European Research Council (ERC) under the European Union’s Horizon 2020 research and innovation programme grant agreement number 819416, and by a Swarnajayanti Fellowship (No. DST/SJF/MSA01/2017-18).}}

\author{Akanksha Agrawal\thanks{Indian Institute of Technology Madras, Chennai, India, \href{mailto:akanksha@cse.iitm.ac.in}{akanksha@cse.iitm.ac.in}.}
\and Paloma T. Lima\thanks{IT University of Copenhagen, Copenhagen, Denmark, \href{mailto:palt@itu.dk}{palt@itu.dk}.}
\and Daniel Lokshtanov\thanks{University of California, Santa Barbara, USA \href{mailto:daniello@ucsb.edu}{daniello@ucsb.edu}.}
\and Pawe\l{} Rz\k{a}\.{z}ewski\thanks{
Warsaw University of Technology, Warsaw, Poland and 
Institute of Informatics, University of Warsaw, Warsaw, Poland \href{mailto:pawel.rzazewski@pw.edu.pl}{pawel.rzazewski@pw.edu.pl}.}
\and Saket Saurabh\thanks{Institute of Mathematical Sciences, Chennai, India and University of Bergen, Norway, \href{mailto:saket@imsc.res.in}{saket@imsc.res.in}.}
\and Roohani Sharma\thanks{\href{mailto:roohani.sharma90@gmail.com}{roohani.sharma90@gmail.com}.}
}

\date{}

\begin{document}
\maketitle

\begin{abstract}
An \emph{independent set} in a graph $G$ is a set of pairwise non-adjacent vertices. A graph $G$ is \emph{bipartite} if its vertex set can be partitioned into two independent sets. In the {\sc Odd Cycle Transversal} problem, the input is a graph $G$ along with a weight function $\w$ associating a rational weight with each vertex, and the task is to find a smallest weight vertex subset $S$ in $G$ such that $G - S$ is bipartite; the weight of $S$, $\w(S) = \sum_{v\in S}\w(v)$. 
We show that {\sc Odd Cycle Transversal} is polynomial-time solvable on graphs excluding $P_5$ (a path on five vertices) as an induced subgraph. 
The problem was previously known to be polynomial-time solvable on $P_4$-free graphs and \textsf{NP}-hard on $P_6$-free graphs~[Dabrowski, Feghali, Johnson, Paesani, Paulusma and Rzążewski, Algorithmica 2020]. Bonamy, Dabrowski, Feghali, Johnson and Paulusma~[Algorithmica 2019] posed the existence of a polynomial-time algorithm on $P_5$-free graphs as an open problem, this was later re-stated by {Rz\k{a}\.{z}ewski} [Dagstuhl Reports, 9(6): 2019] and by Chudnovsky, King, Pilipczuk, {Rz\k{a}\.{z}ewski}, and  Spirkl  [SIDMA 2021], who gave an algorithm with running time $n^{O(\sqrt{n})}$. 
\end{abstract}

\section{Introduction}\label{sec:intro}
In a vertex deletion problem, the input is a graph $G$ along with a weight function $\w: V(G) \rightarrow \mathbb{Q}$, and the task is to find a smallest weight vertex subset $S$ such that removing $S$ from $G$ results in a graph that belongs to a certain target class of graphs; the weight of $S$ being defined as $\w(S) = \sum_{v\in S} \w(v)$. By varying the target graph class we can obtain many classic graph problems, such as deletion to edge-less graphs ({\sc Vertex Cover}), deletion to acyclic graphs ({\sc Feedback Vertex Set}), deletion to bipartite graphs ({\sc Odd Cycle Transversal}), or deletion to planar graphs ({\sc Planarization}). 
With the exception of the class of all graphs, for every target class that contains an infinite set of graphs and is closed under vertex deletion, the vertex deletion problem to that graph class is \textsf{NP}-hard~\cite{LewisYannakakis}. For this reason a substantial research effort has been dedicated to understanding the computational complexity of various vertex deletion problems when the input graph $G$ is required to belong to a restricted graph class as well (see~\cite{brandstadt1999graph} and the companion website~\cite{GCWebpage}).

In this paper, we consider the {\sc Odd Cycle Transversal} ({\sc OCT}) problem, that is the vertex deletion problem to {\em bipartite} graphs. A vertex set $S$ is {\em independent} if no edge has both its endpoints in $S$ and a graph $G$ is bipartite if its vertex set can be partitioned into two independent sets. A graph is bipartite if and only if it has no odd cycles~\cite{diestelBook}. The {\sc OCT} problem is very well studied, and has been considered from the perspective of approximation~\cite{DBLP:journals/combinatorica/GoemansW98,DBLP:journals/mp/FioriniHRV07,DBLP:journals/siamdm/KralSS12}, heuristics~\cite{DBLP:journals/jgaa/Huffner09,DBLP:journals/tcs/AkibaI16}, exact exponential time~\cite{DBLP:journals/mst/RamanSS07,DBLP:journals/corr/abs-1807-10277,DBLP:journals/tcs/KhotR02} and parameterized algorithms~\cite{DBLP:journals/talg/KratschW14,DBLP:journals/talg/LokshtanovNRRS14,DBLP:journals/orl/ReedSV04,DBLP:conf/iwoca/LokshtanovSS09,DBLP:conf/fsttcs/LokshtanovSW12,DBLP:conf/iwpec/JansenK11,DBLP:journals/jcss/KolayMRS20,DBLP:conf/iwpec/JacobBDP21,DBLP:conf/soda/IwataOY14,DBLP:conf/soda/KawarabayashiR10}.

From the viewpoint of restricting the input to specific classes of graphs, \OCT~is known to be polynomia-time solvable on permutation graphs~\cite{BrandstadtKratsch1985} and more generally on graphs of bounded mim-width~\cite{Buixuan2013}.
More recently, $H$-free graphs, that is, graph classes defined by one forbidden induced subgraph, received particular attention. 
Chiarelli et al.~\cite{Chiarellietal2018} showed that \OCT~is \textsf{NP}-complete on graphs of small fixed girth and line graphs. 
This implies that if $H$ contains a cycle or a claw, then \OCT~is \textsf{NP}-hard on $H$-free graphs.
Hence, we can restrict ourselves to the case in which $H$ is a linear forest (i.e., each connected component of $H$ is a path). 
For $P_4$-free graphs a polynomial-time algorithm follows directly from the algorithm of Brandst\"{a}dt and Kratsch~\cite{BrandstadtKratsch1985} for permutation graphs. 
Bonamy et al.~\cite{BonamyDFJP19} posed the question of the existence of a polynomial-time algorithm for \OCT\ on $P_k$-free graphs, for all $k \geq 5$. 
Subsequently, Okrasa and Rz\k{a}\.{z}ewski~\cite{OkrasaRzazewski2020} showed that \OCT~is \textsf{NP}-hard on $P_{13}$-free graphs, 
and shortly thereafter Dabrowski et al.~\cite{Dabrowskietal2020} proved that the problem remains \textsf{NP}-hard even on $(P_2+P_5,P_6)$-free graphs. 
On the other hand, \OCT~is known to be polynomial-time solvable on $sP_2$-free graphs~\cite{Chiarellietal2018} and on $(sP_1+P_3)$-free graphs~\cite{Dabrowskietal2020}, for every constant $s \geq 1$.
Thus, prior to our work, the only {\em connected} graph $H$ such that the complexity status of \OCT~on $H$-free graphs remained unknown was the $P_5$.
For this reason, resolving the complexity status of {\sc OCT} on $P_5$-free graphs was posed as an open problem by Rz\k{a}\.{z}ewski~\cite{Dagstuhl2019}, and by Chudnovsky et al.~\cite{Chudnovskyetal2021}, who gave an algorithm for {\sc OCT} on $n$-vertex $P_5$-free graphs with running time $n^{O(\sqrt{n})}$. 
In this paper, we resolve the open case regarding the computational complexity of {\sc OCT} on $P_5$-free graphs. Specifically, we prove the following theorem.

\begin{restatable}{theorem}{mainthm}\label{thm:oct}
\oct\ on $P_5$-free graphs is polynomial-time solvable.
\end{restatable}

Note that \OCT\ problem can also be rephrased as the problem of finding a maximum weight bipartite (induced) subgraph in the input. The proof of Theorem~\ref{thm:oct} has two main steps as described below.

\noindent{\bf Small covering family for a solution.}
Our main technical contribution towards the proof of Theorem~\ref{thm:oct} is Lemma~\ref{lem:main}. We show that if $G$ is a $P_5$-free graph, then in polynomial time one can construct an $\mathcal{O}(n^6)$-sized family $\mathcal{C}$ of bipartite sets of $G$, such that there exists a maximum weight bipartite subgraph of $G$ which is obtained by taking the induced packing of some of the sets in the family $\mathcal{C}$. Here induced packing means that no two sets intersect, or have an edge between them.

\begin{restatable}{lemma}{mainlemma}\label{lem:main}
    Given a $P_5$-free graph $G$ and a weight function $\w : V(G) \to \mathbb{Q}$, there exists a polynomial-time algorithm that outputs a collection $\mathcal{C} \subseteq 2^{V(G)}$
    of size $\OO(n^6)$ such that:
    \begin{enumerate}
        \item for each $C \in \mathcal{C}$, $G[C]$ is bipartite, and
        \item there exists a set $S \subseteq V(G)$ such that $G[S]$ is bipartite and $\w(S)$ is maximum such that $S = \bigcup_{C \in \mathcal{C}'} C$, where $\mathcal{C}' \subseteq \mathcal{C}$ and for each $C_1, C_2 \in \mathcal{C}'$, $C_1 \cap C_2 =\emptyset$, and $E(C_1,C_2)=\emptyset$.
    \end{enumerate}
\end{restatable}

 The proof of Lemma~\ref{lem:main} heavily relies on the structure of $P_5$-free graphs, including the classical result that every connected $P_5$-free graph has a dominating $P_3$ or dominating clique~\cite{bacso1990dominating}, as well as a somewhat surprising application of the concept of modules~\cite{gallai1967transitiv} (a module in a graph $G$ is a set $M$ of vertices such that every vertex outside of $M$ either is adjacent to all of $M$ or to none of $M$). The formal proof of Lemma~\ref{lem:main} appears in Section~\ref{sec:key}.

\smallskip

\noindent{\bf Reduction to maximum weight independent set in $P_5$-free graphs.} Equipped with Lemma~\ref{lem:main}, we are now looking for a maximum weight bipartite subgraph, each of whose connected components is contained in a given polynomial-sized family $\mathcal{C}$ of connected, bipartite sets. One can now reduce the problem of finding a maximum weight bipartite subgraph in the input graph $G$ to the problem of finding a maximum weight {\em independent set} in an auxiliary graph defined as follows: there is a vertex for each set in $\mathcal{C}$,  an edge between a pair of vertices if and only if the corresponding sets are not disjoint or have an edge connecting them, and the weight of a vertex is the sum of the weights of the vertices in the set that it corresponds to. 
As observed in~\cite{Gartlandetal2021}, if $G$ is $P_5$-free, then so is this auxiliary graph. Therefore, the problem actually reduces to the problem of finding a maximum weight independent set in $P_5$-free graphs, which has been shown to be polynomial-time solvable by Lokshtanov, Vatshelle, and Villanger~\cite{DBLP:conf/soda/LokshantovVV14}. The formal details of this step appears in Section~\ref{sec:algo}.

\section{Preliminaries}\label{sec:prelims}
We denote the set of natural numbers and the set of rational numbers by $\mathbb{N}=\{1,2, \cdots\}$
and $\mathbb{Q}$, respectively, and let $\mathbb{Q}_{\geq 0} = \{x\in \mathbb{Q}\mid x\geq 0\}$.
For $c\in \mathbb{N}$, $[c]$ denotes the set $\{1, \cdots, c\}$.
For a set $X$, we denote by $2^X$ the collection of all subsets of $X$, by ${X \choose i}$ the collection of all $i$-sized subsets of $X$, by ${X \choose \leq i}$ all subsets of $X$ of size at most $i$, and by ${X \choose j_1 \leq i \leq j_2}$ the collection of all $i$-sized subsets of $X$ where $j_1 \leq i \leq j_2$.

Consider a graph $G$. 
For $v \in V(G)$, $N_G(v)$ denotes the set of neighbors of $v$ in $G$, and $N_G[v] = N_G(v) \cup \{v\}$. 
For $X\subseteq V(G)$, $N_G(X) = (\bigcup_{x\in X} N_G(x)) \setminus X$ and $N_G[X] = N_G(X)\cup X$. 
For a subgraph $H$ of $G$, we sometimes write $N_G(H)$ (resp.~$N_G[H]$) as a shorthand for $N_G(V(H))$ (resp.~$N_G[V(H)]$). 
For any $X,Y \subseteq V(G)$, $E_G(X,Y)$ denotes the set of edges of $G$ with one endpoint in $X$ and the other in $Y$. 
Whenever the graph $G$ is clear from the context, we drop the subscript $G$ from the above notations. For any $X \subseteq V(G)$, $G[X]$ denotes the graph induced by $X$, i.e., $V(G[X]) = X$ and $E(G[X])=\{(u,v) \in E(G) \mid u,v \in X\}$. Moreover, by $G-X$ we denote the graph $G[V(G)\setminus X]$. 
For any $v \in V(G)$, we use $G-v$ as a shorthand for $G-\{v\}$. For two graphs $G_1$ and $G_2$ by $G_1\cup G_2$ we denote the graph with vertex set $V(G_1) \cup V(G_2)$ and edge set $E(G_1) \cup E(G_2)$.

For any $i \in \mathbb{N}$, $P_i$ is a path on $i$ vertices and 
a $P_i$-free graph is a graph that has no induced subgraph isomorphic to $P_i$. A {\em dominating set} in a graph $G$ is a set of vertices $D\subseteq V(G)$ such that $N[D] = V(G)$.

\begin{proposition}[Theorem~$8$, \cite{bacso1990dominating}]\label{prop:dom}
Every connected $P_5$-free graph $G$ has a dominating set $D$ such that $G[D]$ is either a $P_3$ or a clique.
\end{proposition}

As noted earlier, the \OCT\ problem can be restated as a maximization problem. We define it directly as a maximization problem below.

\defproblemout{\maxbipfull\ (\OCT)}{An undirected graph $G$ on $n$ vertices and a weight function $\w: V(G) \to \mathbb{Q}$.}{Find $S \subseteq V(G)$, such that $G[S]$ is bipartite and $\w(S)=\sum_{v \in S} \w(v)$ is maximized.}

A set $S \subseteq V(G)$ such that $G[S]$ is bipartite, is called a {\em solution} of $G$. If $\w(S)$ is maximum then we call it an {\em optimal solution} of $(G,\w)$. We will assume that the weight of each vertex is positive, as otherwise, we can remove the vertices with non-positive weight without changing the optimal solution.

\section{Finding a small covering family of the connected components of a solution}\label{sec:key}

The goal of this section is to prove Lemma~\ref{lem:main}.

\mainlemma*

At the heart of our proof for the above lemma lies a property that we prove: for any $S\subseteq V(G)$ where $G[S]$ is a bipartite graph and $C$ is the connected component of $G[S]$,
there is a small set $\widetilde{D} \subseteq V(G)$, which can be guessed efficiently, such that after doing some appropriate cleaning operation of the graph, 
$N[\widetilde{D}] = N[C]$. 
With such a set $\widetilde{D}$ at hand,
we find a replacement for $C$ in $S$ and put this replacement in the family $\mathcal{C}$: by a replacement for $C$ we mean a set $C'$ such that $S \setminus C \cup C'$ is also a bipartite set of weight as large as $S$.
Such a replacement $C'$ is found by exploiting the known algorithm for computing independent sets on $P_5$-free graphs. We will now intuitively explain the steps of our algorithm, and we give a concrete pseudo-code for it in Algorithm~\ref{alg:key}.

In the following, we denote the input graph by $G'$ (instead of $G$). 
Consider a set $S \subseteq V(G')$
where $G'[S]$ is bipartite, and let $C$ be an arbitrary connected component of $G'[S]$.
Our goal is to compute a polynomial-sized family $\mathcal{C}$ of vertex subsets that either contains $C$ itself or a replacement for $C$.
 
To cover the trivial case when $C$ has exactly one vertex, we add the sets $\{v\}$ for each $v \in V(G)$ to the family $\mathcal{C}$ in Line~\ref{line:init}.
Hereafter, assume that $C$ has at least $2$ vertices. As $C$ is a connected, $P_5$-free and bipartite graph, from Proposition~\ref{prop:dom}, there must exist a dominating induced $P_3$ or a dominating $P_2$ for $C$; let one such $P_3$ or $P_2$ be $D$. Note that $V(C) \subseteq N[D]$. The loop at Line~\ref{line:for-one} of Algorithm~\ref{alg:key} will precisely be iterating over such potential $D$s. We will next explain our operations for this fixed $D$, and since we will delete some vertices from our graph, we initialize $G= G'$ at Line~\ref{line:init-graph-G}. We remark that, at all point of time we will implicitly maintain that $S \subseteq V(G)$, and thus maintain that $C$ is the connected component of $G[S]$.

\begin{figure}[t]
    \centering
    \includegraphics[width=.7\textwidth]{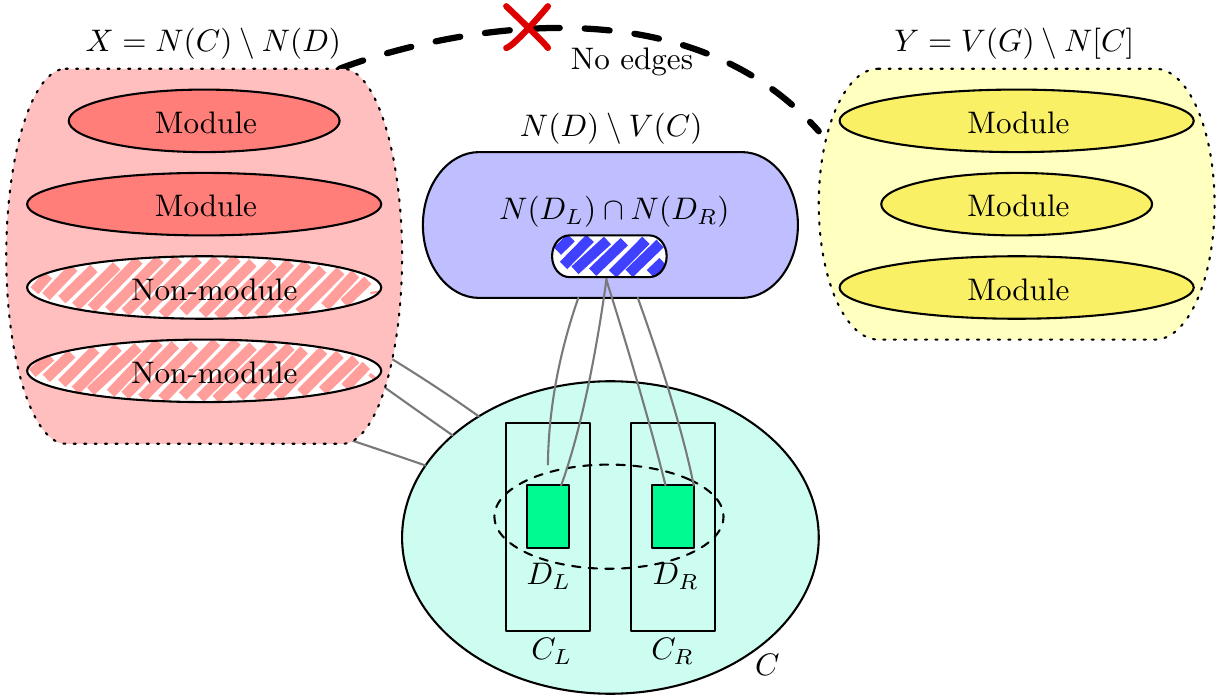}
    \caption{Illustration of various sets and connected components used in the algorithm. The striped ovals denote removal of the corresponding vertices by the algorithm.}
    \label{fig:lemmaStructure}
\end{figure}

We let $C_L\uplus C_R$ be the unique bipartition of $C$ and $D_L \uplus D_R$ be the bipartition of $D$ such that $D_L\subseteq C_L$ and $D_R\subseteq C_R$ (see Figure~\ref{fig:lemmaStructure}).\footnote{Note that both $C$ and $D$ are connected bipartite graphs. 
Thus, upto switching parts, there is exactly one valid bipartition of each of these graphs.} 
Recall that $C$ is the connected component of $G[S]$. 
Notice that any vertex $u \in N(D_L) \cap N(D_R)$ must lie outside of $S$.\footnote{Unless stated explicitly, inside the loop starting at Line~\ref{line:for-one}, we will only be concerned with the graph $G$, and thus, we omit the graph subscripts from the notations like $N(D_L)$.} 
Thus, Line~\ref{line:while-one}-\ref{line:while-one-end} removes such vertices from $G$.  

A set of vertices $A \subseteq V(G)$ is a {\em module} in $G$, if for every pair $x,y \in A$, $N(x) \setminus A = N(y) \setminus A$. Note that checking whether a set $A \subseteq V(G)$ is a module can be done in polynomial time.

We let $X$ be the neighbors of $C$ outside $N(D)$, i.e., $X = N(C) \setminus N(D)$. Also, we let $Y$ be the vertices outside of $C$ and its neighborhood, i.e., $Y = V(G) \setminus N[C]$. We will establish the following statements:
\begin{description}
\item {\em (In Claim~\ref{claim:no-across-edges})} There are no edges between a vertex in $X$ and a vertex in $Y$, i.e., $E(X,Y) = \emptyset$. 

\item {\em (In Claim~\ref{claim:module})} Each connected component of $G[Y]$ is a module in $G$.

\item {\em (In Claim~\ref{claim:ds-c-and-x}, stated roughly here)} the graph $G[X \cup V(C)]$ has a dominating set $\widetilde{D}$ such that: $D\subseteq \widetilde{D}$ and $|\widetilde{D}\setminus D| \leq 3$.
\end{description}

Recall our assumption that $V(C) \subseteq N[D]$.
Using the first property that we establish, we can obtain that for any connected component $Z$ of $G-N[D]$, either $V(Z) \subseteq X$ or $V(Z) \subseteq Y$. 
Now using the second property stated above, we can obtain that the vertices in a non-module connected component of $G-N[D]$ must belong to $X$, and thus it cannot contain a vertex from $S$. This leads us to Line~\ref{line:while-two}-\ref{line:while-two-end} of our algorithm, where we remove vertices from all such non-module connected components. We let $R \subseteq X$ be the vertices remaining after the above stated deletion of vertices from $X$. 

The third property will be used to completely identify the vertices of $N[C]$, however, we may not be able to precisely distinguish which vertices from $N[C]$ belong to $C$. To this end, we will iterate over the potential choices of $D'$ of size at most $3$, so that $\widetilde{D} = D \cup D'$ is a dominating set for $G[R \cup V(C)]$, at Line~\ref{line:for-three}. We will argue that for such a correct $\widetilde{D}$, we must have $N[\widetilde{D}] = N[C]$. Thus, knowing $\widetilde{D}$ precisely gives us $N[C]$, and at Line~\ref{line:supersets} we denote the respective set by $C^*_{D,D'} = N[\widetilde{D}] = N[C]$. 

Recall that $C$ is a connected biparitite induced subgraph of $G[S]$, and any vertex from $G[N[\widetilde{D}]] = G[N[C]]$ that has a neighbor from $V(G') \setminus N[\widetilde{D}]$ cannot belong to $S$ (and thus $C$). Thus, at Line~\ref{line:while-four}-\ref{line:while-four-end} we do the cleaning operation by removing such vertices from $G$. 

After this (in the graph resulting after the previous deletions), we  obtain that $N[C]$ is a connected component of $G$ (not necessarily bipartite) and $S \subseteq V(G)$. 
Now instead of finding $V(C)$ inside $N[C]$ exactly, we will find some $C'$ which will be as good as $V(C)$ as follows. We recall that $D$ is a connected bipartite dominating set for $C$ and thus (upto switching parts), $D$ has a unique bipartition $D_L \uplus D_R$, which we have fixed at Line~\ref{line:for-two}, and we have assumed that $D_L \subseteq C_L$ and $D_R \subseteq C_R$. Due to the cleaning operation of $G$ at Line~\ref{line:while-one}-\ref{line:while-one-end}, $D_L$ and $D_R$ have no common neighbors, i.e., $N[D_L] \cap N[D_R] = \emptyset$. Moreover, as $D$ is a dominating set for the bipartite connected graph $C$, it must be the case that $C_L \subseteq N[D_R]$ and $C_R \subseteq N[D_L]$. Also, recall due to the operations at Line~\ref{line:while-four}-\ref{line:while-four-end} we will be able to conclude that $N[C]$ is a connected component in $G$ and $S$ is an induced subgraph of $G$. Now instead of finding $C_L$ and $C_R$ precisely, we find maximum weight independent sets $I_L$ and $I_R$ on graphs $G[N[D_R]]$ and $G[N[D_L]]$, respectively. Notice that due to our discussions above, $I_L$ and $I_R$ must be disjoint. We then show that $C_{D,D'} = I_L \uplus I_R$ serves as a replacement for $V(C)$, and thus we add $C_{D,D'}$ to $\mathcal{C}$ at Line~\ref{line:update}.

We next state a known result regarding computation of independent sets on $P_5$-free graphs. 

\begin{proposition}[\cite{DBLP:conf/soda/LokshantovVV14}, Independent Set on $P_5$-free]\label{prop:is}
    Given a graph $G$ and a weight function $\w:V(G) \to \mathbb{Q}_{\geq 0}$, there is a polynomial-time algorithm that outputs a set $I \subseteq V(G)$ such that $I$ is an independent set in $G$ and $\w(I)=\sum_{u \in I}\w(u)$ is the maximum.
\end{proposition}

\begin{algorithm}[t] 
\caption{Isolating a connected Component} 
\label{alg:key} 
\begin{algorithmic}[1] 
    \REQUIRE An undirected graph $G'$, a vertex $v \in V(G')$ and a weight function $\w : V(G') \to \mathbb{Q}_{\geq 0}$
    \ENSURE $\mathcal{C} \subseteq 2^{V(G')}$ satisfying the properties of Lemma~\ref{lem:main}
    
    \STATE Initialize $\mathcal{C}=\{\{v\} : v \in V(G) \}$.\label{line:init}

    \FORALL{$D \subseteq {V(G') \choose {2 \leq i \leq 3}}$, where $G'[D]$ is connected and bipartite}\label{line:for-one}
    	\STATE Initialize $G=G'$.\label{line:init-graph-G}
	\STATE Fix a bipartition $D=D_L \uplus D_R$.\label{line:for-two}
            \WHILE{there exists $u \in N(D_L) \cap N(D_R)$}\label{line:while-one}
                  \STATE Delete $u$ from $G$. That is, $G=G-u$.
            \ENDWHILE\label{line:while-one-end}
            \WHILE{there exists a connected component $Z$ of $G-N[D]$ such that $V(Z)$ is not a module in $G$}\label{line:while-two}
                 \STATE Delete $Z$ from $G$. That is, $G=G- V(Z)$.
            \ENDWHILE\label{line:while-two-end}

            \FORALL{$D' \subseteq {V(G) \choose \leq 3}$}\label{line:for-three}
                \STATE Let $C_{D,D'}^{*} = N[D \cup D']$.\label{line:supersets}
                \WHILE{there exists $u \in C_{D,D'}^{*}$ such that $N_{G'}(u) \setminus C_{D,D'}^{*} \neq \emptyset$}\label{line:while-four}
                      \STATE Delete $u$ from $G$. That is, $G=G-u$.
                \ENDWHILE\label{line:while-four-end}
                \STATE Let $I_L$ be a maximum weight independent set obtained by running the algorithm of Proposition~\ref{prop:is} on input $(G[N(D_R)],\w)$.\label{line:ind-one}
                \STATE Let $I_R$ be a maximum weight independent set obtained by running the algorithm of Proposition~\ref{prop:is} on input $(G[N(D_L)],\w)$.\label{line:ind-two}
                \STATE Let $C_{D,D'}=I_L \cup I_R$.\label{line:union}
                \STATE Update $\mathcal{C}=\mathcal{C} \cup \{C_{D,D'}\}$. \label{line:update}
            \ENDFOR
    \ENDFOR
\end{algorithmic}
\end{algorithm}

We remark that for our algorithm we need the algorithm of Proposition~\ref{prop:is} to return a maximum weight independent set and not just the weight such a set. The arguments in~\cite{DBLP:conf/soda/LokshantovVV14} suffice to output such a set.
The following observation is immediate from the description of the algorithm and Proposition~\ref{prop:is}.

\begin{observation}\label{lem:key-running-time}
Algorithm~\ref{alg:key} runs in polynomial time.
\end{observation}

We prove the correctness of our algorithm in the next lemma.

\begin{lemma}\label{lem:key-correctness}
    The family $\mathcal{C}$ outputted by Algorithm~\ref{alg:key} satisfies the properties of Lemma~\ref{lem:main}.
\end{lemma}

\begin{proof}
    To bound the size of the family $\mathcal{C}$, observe that $\mathcal{C}$ is initialized in Line~\ref{line:init} and updated only in Line~\ref{line:update}. At Line~\ref{line:init} the size of $\mathcal{C}$ is $n$ and Line~\ref{line:update} is executed at most $O(n^6)$ times: inside the for-loop at Line~\ref{line:for-one} which contains the for-loop at Line~\ref{line:for-three}. 

    It is also easy to see that the sets of $\mathcal{C}$ induce bipartite subgraphs of $G$.
    Indeed, at Line~\ref{line:init} we add singleton sets,
    and the sets $C_{D,D'}$ added at Line~\ref{line:update} are the union of two independent sets $I_L$ and $I_R$ (Line~\ref{line:union}),
    and therefore $G[C_{D,D'}]$ is bipartite.
        
    We will now show the second property of the lemma. Let $S \subseteq V(G')$ such that $G'[S]$ is bipartite, $\w(S)$ is maximum and $S$ has the maximum number of connected components from $\mathcal{C}$. 
    If all connected components of $S$ are in $\mathcal{C}$ then we are done. Otherwise let $C$ be a connected component of $G'[S]$ such that $C \not \in \mathcal{C}$.
    If $C$ has exactly one vertex
    then $C \in \mathcal{C}$ from Line~\ref{line:init}, and hence a contradiction.
    Therefore assume that $|V(C)| \geq 2$. As $C$ is a connected, $P_5$-free and bipartite graph on at least $2$ vertices (therefore triangle-free, from Proposition~\ref{prop:dom} there exists a dominating set $D$ of $C$ which either induces a $P_3$ or a $P_2$. We consider the execution of the for-loop at Line~\ref{line:for-one} for this $D$, and let $D_L \uplus D_R$ be the bipartition of $D$ considered at Line~\ref{line:for-two} of Algorithm~\ref{alg:key}. Furthermore, let $C_L\uplus C_R$ be the bipartition of $C$ such that $D_L\subseteq C_L$ and $D_R\subseteq C_R$.

    \noindent{\bf [Lines~\ref{line:while-one}-\ref{line:while-one-end}]} Note that any vertex $u \in N(D_L) \cap N(D_R)$ is not in $C$. Furthermore,
    such a vertex $u \not \in S$ 
    because $u \in N(C)$ (since $D \subseteq V(C)$). Thus, $S$ is an induced bipartite subgraph of $G$. 
    
    \noindent{\bf [Lines~\ref{line:while-two}-\ref{line:while-two-end}]} We will now show that any connected component of $G-N[D]$ which is not a module, is a subset of $N(C)$. Once this is proved, $S$ also induces a bipartite subgraph in $G- V(Z)$, where $Z$ is a connected component of $G-N[D]$ that is not a module. Let $X=N(C) \setminus N(D)$ and let $Y= V(G) \setminus N[C]$. 
    
                \begin{figure}[t]
    \centering
    \includegraphics[width=.7\textwidth]{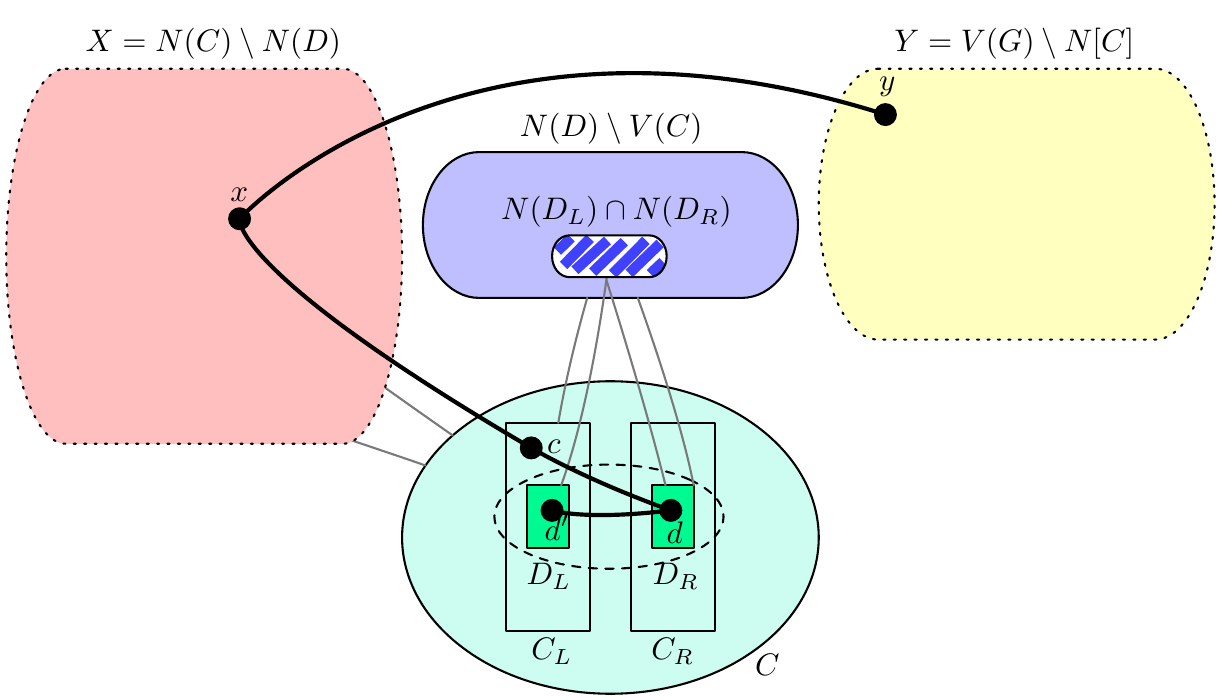}
    \caption{Illustration of various elements in the proof of Claim~\ref{claim:no-across-edges}.}
    \label{fig:fig-X-Y-no-edge}
\end{figure} 

    \begin{claim}\label{claim:no-across-edges}
    $E(X,Y) = \emptyset$.
    \end{claim}
    \begin{proof}
    Suppose for the sake of contradiction that there exists $y \in Y$ and $x \in X$ such that $(y,x) \in E(G)$ (see Figure~\ref{fig:fig-X-Y-no-edge}). Since $x \in X$ and $X = N(C)\setminus N(D)$, there exists $c \in C$ such that $(x,c) \in E(G)$. Also $c \not \in D$, as otherwise $x \in N(D)$. Without loss of generality, say $c \in C_L \setminus D$ (the other case is symmetric). 
    Since $D$ is a dominating set of $C$, 
    there exists $d \in D_R$, 
    such that $(c,d) \in E(G)$.
    Let $d' \in D_L$, and note that $(d,d') \in E(G)$.

    Consider the path $P^*=(y,x,c,d, d')$. We claim $P^*$ is an induced $P_5$. 
    First observe that all the vertices of $P^*$ except $y,x$ are in $C$. 
    Also $E(Y,C) =\emptyset$ because $Y \cap N(C) = \emptyset$ by the definition of $Y$. In particular, $(y,c), (y,d), (y,d') \not \in E(G)$. Since $x \in X$ and $X \cap N(D) = \emptyset$, $(x,d), (x,d') \not \in E(G)$. Finally, since $c,d' \in C_L$ and $C_L$ is an independent set (because it is one of the parts of the bipartition of $C$), $(c,d') \not \in E(G)$.
\end{proof}
From Claim~\ref{claim:no-across-edges}, for any connected component $Z$ of $G-N[D]$, either $V(Z) \subseteq X$, or $V(Z) \subseteq Y$.

        \begin{figure}[t]
    \centering
    \includegraphics[width=.7\textwidth]{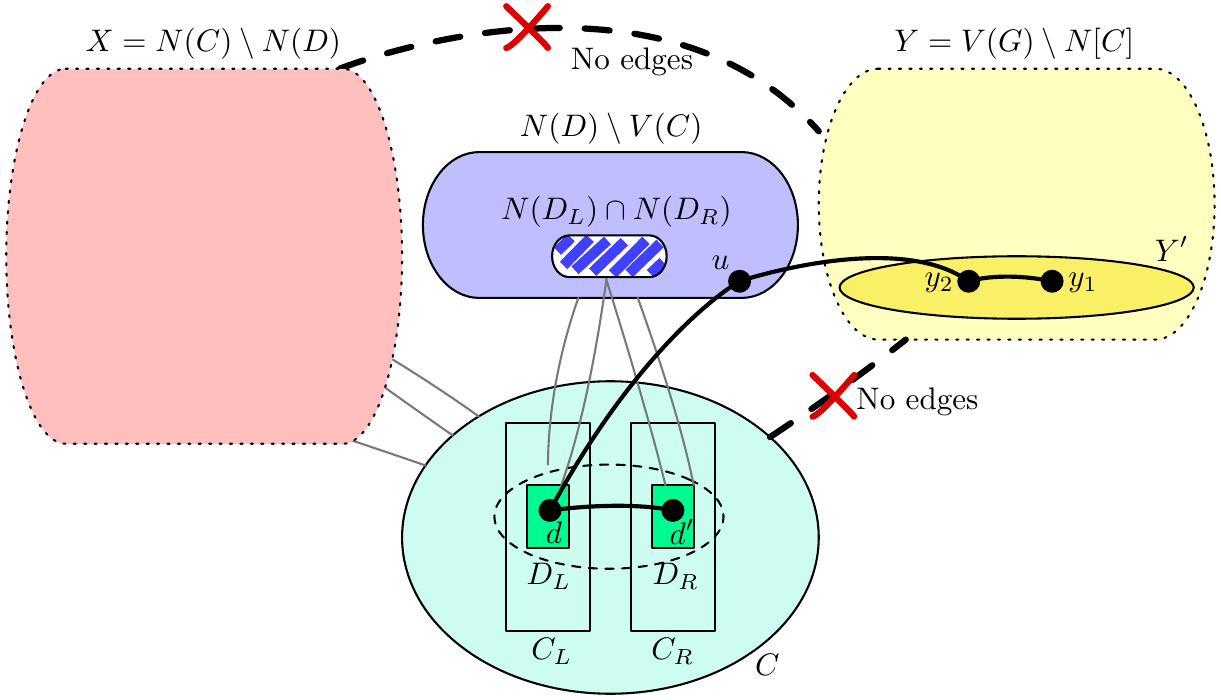}
    \caption{Illustration of various elements in the proof of Claim~\ref{claim:module}.}
    \label{fig:fig-module}
\end{figure}
    
    \begin{claim}\label{claim:module}
    Let $Y'$ be a connected component of $G[Y]$. Then $V(Y')$ is a module in $G$.
    \end{claim}
\begin{proof}
    First note that, from Claim~\ref{claim:no-across-edges}, $N_G(Y) \subseteq N(D)\setminus V(C)$. 
   For the sake of contradiction, say $V(Y')$ is not a module, and thus there exists $y_1,y_2 \in V(Y')$ and $u \in N(D) \setminus V(C)$ such that $(y_1,u) \not \in E(G)$ and $(y_1,y_2), (y_2,u) \in E(G)$ (see Figure~\ref{fig:fig-module}). Since $u \in N(D) \setminus V(C)$ and $D \subseteq V(C)$, 
   there exists $d \in D$ such that $(u,d) \in E(G)$. Without loss of generality let $d \in D_L$. Let $d' \in D_R$ (and note that $(d,d') \in E(G)$). Then $(u,d') \not \in E(G)$ as otherwise $u \in N(D_L) \cap N(D_R)$, which is a contradiction as such vertices do not exist (from Line~\ref{line:while-one}-\ref{line:while-one-end}). Then $P^*=(y_1, y_2, u, d, d')$ is an induced $P_5$ in $G$, which is a contradiction.
\end{proof}
From Claim~\ref{claim:module} if a connected component $Z$ of $G-N[D]$ is not a module, then $V(Z) \subseteq X$. Since $X \subseteq N(C)$, $S$ is also an induced subgraph of $G$ after the execution of Line~\ref{line:while-two}-\ref{line:while-two-end}. 

\noindent{{\bf [Line~\ref{line:for-three}]}} Let $R =N(C) \setminus N(D)$, i.e., $R$ are the remaining vertices of $N(C)$ that are not in $N(D)$ (note that $R$ are precisely the vertices of $X$ from the previous discussion that are not deleted at Line~\ref{line:while-two}-\ref{line:while-two-end}). We now show that there exists a small dominating set of $G[R \cup V(C)]$.

\begin{claim}\label{claim:ds-c-and-x}
$G[R \cup V(C)]$ has a dominating set $\widetilde{D}$ of size at most $6$ such that $D\subseteq \widetilde{D}$ and $|\widetilde{D} \setminus D| \leq 3$. 
\end{claim}
\begin{proof}
    Since $G[R \cup V(C)]$ is a connected and $P_5$-free graph, 
    from Proposition~\ref{prop:dom}, there exists a dominating set of $G[R \cup V(C)]$ which is either a clique or an induced $P_3$. If $D$ itself is a dominating set of $G[R \cup V(C)]$, then the claim trivially follows. Otherwise, we consider a dominating clique or a dominating induced $P_3$, $D'$ of $G[R \cup V(C)]$ of minimum possible size. If $D'$ induces a $P_3$ then, $\widetilde{D} = D'\cup D$ satisfies the requirement of the claim. Now we consider the case when $D'$ is a clique. Using $D'$ we will construct a dominating set of $G[R \cup V(C)]$ with at most $6$ vertices, containing the vertices from $D$. Intuitively speaking, apart from $D$ (whose size is at most $3$) we will add vertices from $D'\cap V(C)$ and at most one more vertex. We remark that as $D'$ is a clique and $C$ is a bipartite graph, $|D'\cap V(C)| \leq 2$. 
 
    Let $R_1, \ldots, R_p$ be the connected components of $G[R]$. 
    Since $D'$ is a clique, $D'$ intersects at most one $R_i$, i.e., there exists an $i \in [p]$, such that $D' \cap V(R_j) =\emptyset$, for all $j \in [p]\setminus \{i\}$.
    Therefore the vertices of $\bigcup_{j \in [p] \setminus \{i\}}V(R_i)$ are dominated by the vertices of $D' \cap V(C)$. If $D'\cap V(R_i) =\emptyset$, then notice that $D'\subseteq V(C)$, where $|D'|\leq 2$, and thus $\widetilde{D} = D'\cup D$ satisfies the requirement of the claim. Now suppose that $D'\cap V(R_i) \neq \emptyset$, and consider a vertex $x \in D' \cap V(R_i)$. Recall that $x\in N(C)$ and $R_i$ is a module in $G$. Thus, there exists a vertex $v' \in V(C)\cap N(x)$, and moreover, we have $V(R_i)\subseteq N(v')$. 
    Note that $D$ dominates each vertex in $C$, $v'$ dominates each vertex in $R_i$, and $D'\cap V(C)$ dominates each vertex in $\bigcup_{j \in [p] \setminus \{i\}} V(R_i)$. Thus, $\widetilde{D} = D \cup \{v'\} \cup (D' \cap V(C))$ dominates each vertex in $G[R \cup V(C)]$ and $|D'\cap V(C)| \leq 2$. This concludes the proof.  
\end{proof}

Now consider the execution of the for-loop at Line~\ref{line:for-three} when $\widetilde{D} = D'\cup D$ is a dominating set for $G[R \cup C]$. 

\noindent{{\bf [Line~\ref{line:supersets}]}} Our objective is to argue that the set $C^*_{D,D'}$ 
at Line~\ref{line:supersets} is equal to $ N[C]$. To obtain the above, it is enough to show that $N[\widetilde{D}] = N[C]$. To this end, we first obtain that $N[C] \subseteq N[\widetilde{D}]$. Recall that, after removing the vertices from $X$ at Line 8, $N[C] = R \cup V(C) \cup N(D)$. As $\widetilde{D} = D \cup D'$ is a dominating set for $G[R\cup V(C)]$, we have $R \cup V(C) \subseteq N[\widetilde{D}]$. Moreover, as $D\subseteq \widetilde{D}$, we have $N(D) \subseteq N[\widetilde{D}]$. Thus we can conclude that $N[C] \subseteq N[\widetilde{D}]$. We will next argue that $N[\widetilde{D}] \subseteq N[C]$. Recall that from Claim~\ref{claim:no-across-edges}, $E(R,Y) =\emptyset$. Thus, for any vertex $v\in D' \setminus V(C)$, $N(v) \subseteq N[C]$. In the above, when $v\in D' \setminus V(C)$, without loss of generality we can suppose that $v\in R \subseteq N(C)$, as $\widetilde{D} = D \cup D'$ is a dominating set for $G[R\cup V(C)]$. Thus, for each $v\in D' \setminus V(C)$, $N[v] \subseteq N[C]$. Also, $D\subseteq V(C)$, and thus, $N[D]\subseteq N[C]$. Hence it follows that $N[\widetilde{D}] \subseteq N[C]$. Thus we obtain our claim that, $N[\widetilde{D}] = N[C]$. 

\noindent{{\bf [Lines~\ref{line:while-four}-\ref{line:while-four-end}]}}
Since $C^*_{D,D'} = N[C]$, if there exists $u \in C^*_{D,D'}$ such that $u$ has a neighbor outside $C^*_{D,D'}$, then $u \in N(C)$ and hence $u \not \in S$. Thus $S$ induces a bipartite subgraph even in the graph obtained by deleting such vertices. Notice that after execution of these steps, $N[C]$ is a connected component of $G$, as removing a vertex from $N(C)$ cannot disconnect the graph $N[C]$. 

\noindent{{\bf [Lines~\ref{line:ind-one}}-\ref{line:union}]} 
Recall that $D$ is a dominating set of $C$ and $C_L \uplus C_R$ is a bipartition of $C$, where $C_L \subseteq N(D_R)$ and $C_R \subseteq N(D_L)$. Also, from Line~\ref{line:while-one}-\ref{line:while-one-end}, $N(D_L) \cap N(D_R) = \emptyset$. Let $I_L$ (resp.~$I_R$) be the maximum weight independent set in $G[N[D_R]]$ (resp.~$G[N[D_L]]$) computed at these steps. Then $\w(I_L) \geq \w(C_L)$ (resp.~$\w(I_R) \geq \w(C_R)$) because $C_L$ (resp.~$C_R$) is an independent set in $G[N[D_R]]$ (resp.~$G[N[D_L]]$). Thus, $\w(C_{D,D'}) \geq \w(V(C))$. 

Let $S'= (S\setminus V(C)) \cup C_{D,D'}$. Then $\w(S') \geq \w(S)$. Also $G[S']$ is bipartite because $G[C_{D,D'}]$ is bipartite and $E(S\setminus V(C), C_{D,D'}) =\emptyset$ because $N[C]$ is a connected component of (the reduced graph) $G$. Additionally $S'$ has more connected components in $\mathcal{C}$ compared to $S$, which contradicts the choice of $S$.
\end{proof}

The proof of Lemma~\ref{lem:main} follows from Observation~\ref{lem:key-running-time} and Lemma~\ref{lem:key-correctness}.

\section{Reduction to {\sc Maximum Weight Independent Set} on $P_5$-free Graphs 
and Proof of Theorem~\ref{thm:oct}
}\label{sec:algo}
In this section we show how, using Lemma~\ref{lem:main}, we can reduce, in polynomial time, the problem of finding a maximum weight induced bipartite subgraph on $P_5$-free graphs, to the problem of finding a maximum weight independent set on $P_5$-free graphs.
Recall that Lokshtanov et al.~\cite{DBLP:conf/soda/LokshantovVV14} gave a polynomial-time algorithm for the latter problem.
Their result, together with the reduction in this section proves Theorem~\ref{thm:oct}.

Recall that $(G,\w)$ is an instance of \OCT\
where $G$ is a $P_5$-free graph.
Let $\mathcal{C}$ be the family of vertex subsets of $G$
 obtained from Lemma~\ref{lem:main} on input $(G,\w)$.
 Let $\widetilde{\mathcal{C}} \subseteq 2^{V(G)}$ be the family containing the vertex set of each connected component in the graph induced by each set in the family $\mathcal{C}$. 
We say that two sets $C_1, C_2 \in \widetilde{\mathcal{C}}$ {\em touch} each other, if either $V(C_1) \cap V(C_2) \neq \emptyset$ or $E(C_1,C_2)\neq \emptyset$.
The following lemma reduces the task of finding a maximum weight bipartite subgraph in $G$ to that of finding a pairwise non-touching collection of sets in $\mathcal{C'}$ of maximum total weight.

\begin{lemma}\label{lem:oct-equiv-packing}
$(G,\w)$ has an induced bipartite subgraph of weight (defined by $\w$) $W$ if and only if there exists $\mathcal{C}' \subseteq \widetilde{\mathcal{C}}$ such that no two sets in $\mathcal{C}'$ touch each other and $\sum_{C \in \mathcal{C}'} \w(C) = W$.
\end{lemma}
\begin{proof}
    From Lemma~\ref{lem:main}, there exists $S \subseteq V(G)$ such that $G[S]$ is bipartite, the weight of $S$ (with respect to $\w$) is maximum, and all the connected components of $G[S]$ are contained in $\widetilde{\mathcal{C}}$. 
    Therefore $S$ corresponds to a pairwise non-touching sub-collection of $\widetilde{\mathcal{C}}$ of total weight equal to $\w(S)$. For the other direction, since every set in $\widetilde{\mathcal{C}}$ is connected and bipartite, we conclude that the union of the sets in any sub-collection of $\widetilde{\mathcal{C}}$, that contains no two sets that touch each other, forms an induced bipartite subgraph of $G$.
\end{proof}

Let $G^{\blob}_{\widetilde{\mathcal{C}}}$ be an auxiliary graph whose vertex set corresponds to sets in $\widetilde{\mathcal{C}}$ and there is an edge between two vertices of $G^{\blob}_{\widetilde{\mathcal{C}}}$ if and only if the corresponding sets touch other. Formally, $V(G^{\blob}_{\widetilde{\mathcal{C}}}) =\{v_C : C \in \widetilde{\mathcal{C}}\}$ and for any $v_C,v_{C'} \in V(G^{\blob}_{\widetilde{\mathcal{C}}})$, there exists $(v_C,v_{C'}) \in E(G^{\blob}_{\widetilde{\mathcal{C}}})$ if and only if $C$ and $C'$ touch each other. Let $\w^{\blob}: V(G^{\blob}_{\widetilde{\mathcal{C}}}) \to \mathbb{Q}$ be defined as follows: $\w^{\blob}(v_C) = \sum_{u \in C} \w(u)$. 
In the vocabulary of Gartland et al.~\cite{Gartlandetal2021}, $G^{\blob}_{\widetilde{\mathcal{C}}}$ is an induced subgraph of \emph{the blob graph of $G$}, and they observe that the blob graph (and thus its every induced subgraph) of a $P_5$-free graph is $P_5$-free.

\begin{proposition}[Theorem~$4.1$,~\cite{Gartlandetal2021}]\label{prop:blob}
    If $G$ is $P_5$-free, then $G^{\blob}_{\widetilde{\mathcal{C}}}$ is also $P_5$-free.
\end{proposition}

The following lemma is an immediate consequence of Lemma~\ref{lem:oct-equiv-packing}.
\begin{lemma}\label{lem:oct-equiv-is}
$(G,\w)$ has an induced bipartite subgraph of weight (with respect to the weight function $\w$) $W$ if and only if $G^{\blob}_{\widetilde{\mathcal{C}}}$ has an independent set of weight $W$, with respect to the weight function $\w^{\blob}$.
\end{lemma}

From Proposition~\ref{prop:blob},  if $G$ is $P_5$-free, then $G^{\blob}_{\widetilde{\mathcal{C}}}$ is also $P_5$-free. 
So by Lemma~\ref{lem:oct-equiv-is}, the problem actually reduces to finding a maximum weight independent set in a $P_5$-free graph, which can be done by~\cite{DBLP:conf/soda/LokshantovVV14}.
Now we are ready to prove our main result, i.e., Theorem~\ref{thm:oct}.

\begin{proof}[Proof of Theorem~\ref{thm:oct}]
Let $(G,\w)$ be an instance of \OCT. Construct an instance $(G^{\blob}_{\widetilde{\mathcal{C}}},\w^{\blob})$ as described earlier.
Note that the number of vertices of $G^{\blob}_{\widetilde{\mathcal{C}}}$ is $|\widetilde{\mathcal{C}}|$, 
which is at most $|\mathcal{C}|$ times the maximum number of connected components of any subgraph induced by an element of $\mathcal{C}$. 
Since $|\mathcal{C}| =\OO(n^6)$ by Lemma~\ref{lem:main}, the number of vertices of $G^{\blob}_{\widetilde{\mathcal{C}}}$ is $\OO(n^7)$. 
Also the construction of this graph takes time polynomial in $n$.

Thus, by Lemma~\ref{lem:oct-equiv-is}, the problem reduces to finding a maximum-weight independent set in a $P_5$-free graph $G^{\blob}_{\widetilde{\mathcal{C}}}$. A maximum-weight independent set on $P_5$-free graphs can be found in polynomial time using the algorithm of~\cite{DBLP:conf/soda/LokshantovVV14} (Proposition~\ref{prop:is}). This concludes the proof of Theorem~\ref{thm:oct}.
\end{proof}

\section{Conclusion}\label{sec:conclusion}
We gave a polynomial-time algorithm for \OCT{} on $P_5$-free graphs. Several interesting problems in this direction remain open.

\begin{enumerate}
    \item The algorithms for {\sc Independent Set} on $P_t$-free graphs~\cite{GartlandLokshtanov2020,DBLP:conf/sosa/PilipczukPR21} also work for counting the number of independent sets of a given size within the same time bound. Because of the ``greedy choice'' implicit in the statement of Lemma~\ref{lem:main}, our algorithm does not work for counting the number of induced bipartite subgraphs of a given size. Does a polynomial (or quasi-polynomial) time algorithm exist for this problem in $P_5$-free graphs?

    \item For every fixed positive integer $k$  we can determine whether $G$ is $k$-colorable in polynomial time on $P_5$-free graphs~\cite{DBLP:journals/algorithmica/HoangKLSS10}. Therefore, in light of Theorem~\ref{thm:oct} it makes sense to ask whether for every positive integer $k$ there exists a polynomial- or quasi-polynomial-time algorithm that takes as input a graph $G$ and outputs a maximum-size (or maximum-weight) set $S$ such that $G[S]$ is $k$-colorable.
    The work of Chudnovsky et al.~\cite{Chudnovskyetal2021} provides a subexponential-time algorithm for this problem (and, actually, its further generalizations).

    \item Our result completes the classification of all {\em connected} graphs $H$ into ones such that \OCT{} on $H$-free graphs is polynomial time solvable or \textsf{NP}-complete. Such a classification for all graphs $H$ (connected or not) remains open. Even more generally, one could hope for a complete classification of the complexity of \OCT{} on ${\cal F}$-free graphs for every \emph{finite} set ${\cal F}$.
\end{enumerate}

\bibliographystyle{alpha}
\bibliography{references}

\end{document}